\documentclass[11pt, letterpaper, USenglish]{article}

\usepackage[utf8]{inputenc}
\usepackage{comment}
\usepackage{graphicx}
\usepackage{float}
\usepackage{hyperref}
\usepackage{xcolor}
\usepackage{xspace}
\usepackage{tikz}
\usepackage{mathtools}

\usepackage{cleveref}
\hypersetup{hidelinks}

\usepackage{amsthm}
\usepackage{amsfonts}
\usepackage{thm-restate}

\usepackage[margin=1in]{geometry}

\usepackage{algorithm}
\usepackage{algpseudocode}
\usepackage{todonotes}

\newcommand{\E}{\mathbb{E}}

\newcommand{\poly}{\mathrm{poly}}

\newcommand{\parallelrulingset}{\textrm{Parallel 2-Ruling-Set}\xspace}

\newcommand{\mpc}{\textrm{MPC}\xspace}
\newcommand{\clique}{\textrm{Congested Clique}\xspace}

\newcommand{\vsetaside}{\ensuremath{V^*}\xspace}
\newcommand{\lubymis}{\ensuremath{I}\xspace}

\newcommand{\gsampled}{\ensuremath{G_\mathrm{samp}}\xspace}
\newcommand{\vsampled}{\ensuremath{V_\mathrm{samp}}\xspace}

\newcommand{\whp}{w.h.p.\xspace}

\newtheorem{theorem}{Theorem}
\newtheorem{definition}[theorem]{Definition}
\newtheorem{lemma}[theorem]{Lemma}

\newtheorem*{theorem*}{Theorem}

\title{Time and Space Optimal Massively Parallel Algorithm for the 2-Ruling Set Problem}

\author{M\'elanie Cambus\footnote{Aalto University, Finland, \{melanie.cambus, shreyas.pai, jara.uitto\}@aalto.fi. M\'elanie and Shreyas are supported by Academy of Finland, Grant 334238} \and Fabian Kuhn\footnote{University of Freiburg, Germany, kuhn@cs.uni-freiburg.de}  \and  Shreyas Pai\footnotemark[1] \and Jara Uitto\footnotemark[1]}

\date{}

\begin{document}

\maketitle
\begin{abstract}
  In this work, we present a constant-round algorithm for the $2$-ruling set problem in the \clique model.
  As a direct consequence, we obtain a constant round algorithm in the \mpc model with linear space-per-machine and optimal total space.
  Our results improve on the $O(\log \log \log n)$-round algorithm by~[HPS, DISC'14] and the $O(\log \log \Delta)$-round algorithm by~[GGKMR, PODC'18].
  Our techniques can also be applied to the semi-streaming model to obtain an $O(1)$-pass algorithm.

  Our main technical contribution is a novel sampling procedure that returns a small subgraph such that \emph{almost} all nodes in the input graph are adjacent to the sampled subgraph.
  An MIS on the sampled subgraph provides a $2$-ruling set for a large fraction of the input graph.
  As a technical challenge, we must handle the remaining part of the graph, which might still be relatively large.
  We overcome this challenge by showing useful structural properties of the remaining graph and show that running our process twice yields a $2$-ruling set of the original input graph with high probability. 
\end{abstract}

\newpage

\section{Introduction}
In this paper, we design and analyze a parallel algorithm for finding ruling sets. For a graph $G = (V, E)$ (with $|V| = n$ and $|E| = m$) and integer $\beta \ge 1$, a $\beta$-ruling set $S \subseteq V$ is a set of non-adjacent nodes such that for each node $u \in V$, there is a ruling set node $v \in S$ within $\beta$ hops. This is a natural generalization of one of the most fundamental problems in parallel and distributed graph algorithms: Maximal Independent Set (MIS), which corresponds to a $1$-ruling set. Ruling sets are closely related to clustering problems like Metric Facility Location, as fast algorithms for $\beta$-ruling sets imply fast algorithms for $O(\beta)$-approximate metric facility location \cite{Berns2012,inamdar2018}.

Our main contribution is an $O(1)$ round algorithm for $2$-ruling sets in \clique, improving on the $O(\log \log \log n)$ time algorithm by~\cite{Hegeman2014} and $O(\log \log \Delta)$ algorithm by~\cite{ghaffari2018improved}, where $\Delta$ denotes the maximum degree of the input graph. While problems like minimum spanning tree \cite{nowicki2021mst} and $(\Delta + 1)$-coloring \cite{czumajdp20coloring} surprisingly admit constant round solutions in the \clique model, the MIS problem and even $\beta$-ruling set for $\beta = O(1)$ have resisted attempts to obtain constant round algorithms. We make significant progress in this direction by giving the first $O(1)$ round algorithm for $2$-ruling sets in \clique. 

The \clique model~\cite{LotkerPPP2003CongClique}, is a distributed synchronous message-passing model where each node is given its incident edges as an input and the nodes perform all-to-all communication in synchronous rounds.
The crucial limitation is that the size of the messages sent between any pair of nodes, in a single round, is limited to $O(\log n)$ bits, where $n$ is the number of nodes in the input graph. The goal is to minimize the number of required synchronous communication rounds.

As a consequence of our \clique algorithm, we obtain a constant round algorithm for $2$-ruling sets in the Linear Memory \mpc model~\cite{Karloff2010, Goodrich2011, Beame2017} and a constant pass algorithm in the semi-streaming model~\cite{Feigenbaum2005, Feigenbaum2005a, Mut2005}. 
The \clique and \mpc implementations are asymptotically optimal in both time and space, as they use constant rounds and $O(n+m)$ total memory. 
In the semi-streaming model, one typically aims for very few passes, ideally just one. While we show that constant passes are enough to solve $2$-ruling sets with $O(n)$ space, discovering the precise number of passes required is an interesting open question. 
The implementation details and definitions of the different models of computation can be found in \Cref{sec: implementation-main}. Our main results are captured in the following theorem.

\begin{theorem*}[Main Theorem]
    There is a randomized parallel algorithm to the $2$-ruling set problem that can be implemented in $O(1)$ rounds/passes (1) in the \clique model, (2) in the \mpc model with $O(n)$ words of local memory and $O(n+m)$ words of total memory, and (3) in the semi-streaming model with $O(n)$ words of space. The running time guarantee of the algorithm holds with high probability (\whp)\footnote{Following the standard, we say that an event holds with high probability if it holds with probability $1 - 1/n^c$ for a constant $c\geq 1$ we can choose.}.
\end{theorem*}

\subsection{Previous Works on Ruling Sets and MIS}

Recall that a $2$-ruling set $S$ is a set of non-adjacent nodes, such that each node in the input graph has a node from $S$ within its $2$-hop neighborhood.
This is a strictly looser requirement than the one of an MIS and hence, an MIS algorithm directly implies a $2$-ruling set algorithm.
The classic algorithms by Luby~\cite{luby86} and Alon, Babai, and Itai~\cite{alon86} yield $O(\log n)$ time algorithms for ruling sets.
In \clique and \mpc, this was later improved to $\widetilde{O}(\sqrt{\log \Delta})$~\cite{Ghaffari2017} and the current state-of-the-art for MIS is $O(\log \log \Delta)$ rounds~\cite{ghaffari2018improved}.

When focusing on the $2$-ruling set problem, roughly a decade ago,~\cite{Berns2012} gave an expected $O(\log \log n)$ time algorithm and~\cite{Hegeman2014} gave an $O(\log \log n)$ time algorithm \whp\ in the \clique model.
Combining the result of~\cite{Hegeman2014} with the $O(\log \log \Delta)$ algorithm for MIS, this was improved to $O(\log \log \log n)$ rounds (\whp).
The $2$-ruling set problem can be solved \emph{deterministically} in $O(\log \log n)$ rounds in the \clique model \cite{Pai2022}. On the other hand, \cite{assadi2022rounds} shows that $\Omega(\log \log n)$ rounds are required to compute an MIS in the Broadcast \clique model.

The $O(\log \log \Delta)$ algorithm by~\cite{ghaffari2018improved} carries over to the semi-streaming model.
We note that there is an older variant of this algorithm that also yields an $O(\log \log \Delta)$-pass randomized \emph{greedy} MIS, which is given by picking a random permutation over the nodes~\cite{ahn2015correlation}. Then the permutation is iterated over and, whenever possible, the current node is added to the MIS. However, this approach requires a polylogarithmic overhead in space and hence, does not work in the \clique model.
While it is known that computing an MIS in a \emph{single-pass} of the stream requires $\Omega(n^2)$ memory \cite{cormode2019MISstreaminglb}, there has been progress towards designing single-pass semi-streaming algorithms for $2$-ruling sets, but the current approaches require significantly larger than $O(n)$ memory: \cite{konrad2019complexity} shows that it can be done in $O(n\sqrt{n})$ memory and this was improved to $O(n^{4/3})$ by \cite{AssadiDDISC21}. The $\Omega(n^2)$-space lower bound for MIS was generalized to $(\alpha, \alpha-1)$-ruling sets by \cite{AssadiDDISC21}. 
An $(\alpha, \beta)$-ruling set is a set $S$ of nodes such that nodes within $S$ are distance at least $\alpha$ apart and every node in $V \setminus S$ has a node in $S$ within distance $\beta$. An MIS is a $(2,1)$-ruling set, and a $2$-ruling set is a $(2,2)$-ruling set, in this paper we drop the $\alpha$ parameter for sake of convenience as it is always $2$.
A single-pass semi-streaming algorithm for $2$-ruling set has not been ruled out.

\subsection{A High-Level Technical Overview of Our Algorithm}

Our strategy is to compute an MIS iteratively on subgraphs of size $O(n)$ until all nodes are covered. We begin with a sampling process where each node $u$ is sampled independently with probability $1/\sqrt{\deg(u)}$. Call the set of sampled vertices \vsampled. 
The intuition behind this sampling probability is to maximize the probability that each node has a sampled neighbor while ensuring $O(n)$ edges in the sampled subgraph $G[\vsampled]$.
To see why the sampled subgraph is small, assume we have a $d$-regular graph. Hence, each node is independently sampled with probability $1/\sqrt{d}$. For an edge $\{u, v\}$ to be sampled, both $u$ and $v$ must be sampled together. Therefore, an edge exists in the sampled graph with probability $1/d$, and we can say that there are at most $O(n)$ edges in the sampled graph in expectation (since the total number of edges is $nd$). 

We show that the same expectation holds for general graphs by orienting the edges based on the degree of their end points and counting the number of outgoing sampled edges per node. In order to show that $G[\vsampled]$ has $O(n)$ edges \whp, we face a technical challenge that the random choices for all edges are not independent. 
The dependencies disallow the use of standard Chernoff bounds, but we overcome the challenge by using the method of bounded differences to show concentration. For convenience, we state the concentration bounds in \Cref{sec: concentration}.

We classify nodes in the graph as either \emph{good} or \emph{bad}, based on how their initial degree compares with the sum of sampling probabilities of their neighbors:
\begin{definition}[Good and Bad Nodes]\label{def:good-bad-nodes}
     A node $u$ is good if $\sum_{v \in N(u)} 1/\sqrt{\deg(v)} \geq \gamma \log \deg(u)$ (for a constant $\gamma$). If a node is not good, it is bad.
\end{definition}

Note that this definition is based entirely on the graph structure (i.e., degrees of all nodes).
Furthermore, it is straightforward to detect which nodes are good and bad in $O(1)$ rounds.

By definition, good nodes are expected to have many sampled neighbors, so they have a relatively high probability of being covered by a node in \vsampled. So we put the good nodes that do not have any sampled neighbors into a set \vsetaside and process it later. We call \vsetaside the nodes that are \emph{set aside}. Therefore, the remaining graph of uncovered nodes only consists of bad nodes. In addition to all the good nodes without a sampled neighbor, we also add some bad nodes that are very likely to be covered (but weren't) into \vsetaside. We use the fact that all nodes in \vsetaside have a good chance of being covered to show that the graph $G[\vsetaside]$ has $O(n)$ edges \whp The main difficulty in proving this claim is that nodes are not put independently in \vsetaside, and only nodes that are a certain distance apart are independent of each other. Moreover, the probability of a single node in \vsetaside with degree $d$ being covered is only $1 - 1/\poly(d)$ which is not enough to do a simple union bound over all nodes. We overcome this challenge by showing that \vsetaside can be partitioned into large sets of far apart nodes, which allows us to (1) boost the probability of each set being covered, and (2) apply a union bound over all sets.

Finally, we do an intricate counting argument that bounds the number of uncovered bad nodes. We do this by counting the number of edges that can exist between bad nodes and higher degree nodes, and show that the number of bad nodes with degree $d$ is roughly bounded by $n/\sqrt{d}$. Note that this bound is not enough to show that the remaining graph is small. So we run the entire sampling and setting nodes aside process again to get the bound to be roughly $n/d$, which makes it straightforward to show that the remaining graph has $O(n)$ edges.

\section{\parallelrulingset}{\label{sec: ruling}}
In this section, we present a parallel algorithm for finding a $2$-ruling set. For technical convenience, we assume that the maximum degree of the input graph is bounded by $n^{\alpha}$ for some constant $\alpha < 1/8$, and in \Cref{sec:degree-reduction} we show how to remove this assumption. When computing an MIS on the sampled vertices, we use a version of the Luby's algorithm that works as follows: in an iteration, each node picks a real number independently and uniformly at random in $\left[0, 1\right]$, local minima join the MIS, and these two steps are repeated until all nodes are either in the MIS or have a neighbor in the MIS.  It is well known that this algorithm terminates in $O(\log n)$ iterations \whp (see for example \cite{ghaffari2016improved}). 

Recalling the definition of good and bad nodes from \Cref{def:good-bad-nodes}, let $B_{d}$ be the set of bad nodes in $G$ with (initial) degree in the range $[d, 2d)$.
We describe the \parallelrulingset algorithm in \Cref{alg:ruling-set}. 

\begin{algorithm}[tbh]
\caption{\parallelrulingset}\label{alg:ruling-set}
\begin{algorithmic}[1]
    \Statex \textbf{Input}: Graph $G = (V, E)$, with $\Delta \le n^{\alpha}$ $(\alpha < 1/8)$. Each node $v \in V$ knows its degree $\deg(v)$ in $G$. 
    \State Each node $v \in G$ is independently sampled into $\vsampled$ with probability $1/\sqrt{\deg(v)}$. \label{line:2rs-sample}
    \State We put good nodes that do not have any neighbors in \vsampled in the set \vsetaside.\label{line:2rs-setaside-1}
    \State Compute an MIS $\lubymis$ on $\gsampled = G[\vsampled]$ using Luby's algorithm first on the bad nodes and then on the rest of the nodes. \label{line:2rs-luby}
    \State If any uncovered $u \in B_d$ has a neighbor $v$ that has at least $c \sqrt{d} \log^5(d)$ neighbors in $B_d$, we put $u$ in $\vsetaside$. \Comment{$B_{d}$ is the set of bad nodes in $G$ with (initial) degree in $[d, 2d)$}\label{line:2rs-setaside-2}
    \State Compute an MIS on $G[\vsetaside]$.\label{line:2rs-setaside-mis}
    \State Run Lines~\ref{line:2rs-sample} to~\ref{line:2rs-setaside-mis} on $G' =  G[V \setminus (\mathcal{I} \cup  N(\mathcal{I}) \cup N(N(\mathcal{I})))]$ where $\mathcal{I}$ is the set of MIS nodes. \label{line:2rs-second-time}
    \State Compute an MIS on the graph induced by the uncovered nodes. \label{line:2rs-mis-residual}
    \Statex \textbf{Output}: The set of nodes that joined an MIS during the algorithm.
\end{algorithmic}
\end{algorithm}

A node is considered covered if and only if it is at most $2$-hops away from a node in the MIS, and two adjacent nodes can never join the MIS. Since the last step of the algorithm computes an MIS on the uncovered nodes, it is guaranteed to output a valid $2$-ruling set, hence proving the following theorem.

\begin{restatable}{theorem}{2rulingsetthm}
\label{thm: 2rulingset}
 The \parallelrulingset algorithm (\Cref{alg:ruling-set}) outputs a valid $2$-ruling set. 
\end{restatable}

\subsection{Structural Properties of the Subgraphs}\label{subsection:2rs-small-graphs}

We will now prove key structural properties of the different subgraphs on which we compute an MIS in \Cref{alg:ruling-set}. This will allow for fast implementation of this algorithm in linear memory \mpc, \clique, and semi-streaming models (see \Cref{sec: implementation-main}).
We first prove the fact that the sampled graph has a linear number of edges \whp

\begin{lemma}\label{lemma: 2rs-gsamp-small}
 The sampled graph $\gsampled$ has $O(n)$ edges \whp
\end{lemma}
\begin{proof}
Let $X$ be the random variable denoting the number of edges in \gsampled. 
Let $X_u$ be the indicator random variable for the event that $u$ is sampled in \vsampled and let $Y_e$ be the indicator random variable for the event that edge $e$ belongs to \gsampled. 
We orient all edges from the end point with lower initial degree to higher initial degree. 
By the degree sum lemma, we have $X = \sum_{u\in V} \sum_{e \in \mathrm{Out}(u)} Y_e$, where $\mathrm{Out}(u)$ is defined as the set of outgoing edges of $u$.

Consider an oriented edge $e = (u, v)$ with $\deg(u) \le \deg(v)$. We have that both $u$ and $v$ are sampled with probability at most $1/\sqrt{\deg(u)}$, so the probability that $e$ is in \gsampled is at most $1/\deg(u)$. Therefore, $\E[Y_e] \le 1/\deg(u)$, and $\E[X] \le n$.

We can interpret $X$ as a function of the random variables $X_u, u \in V$, where changing one coordinate changes $X$ by at most $\Delta = n^{\alpha}$. 
Therefore, $X$ follows the bounded differences property with bounds $c_u = n^{\alpha}$ for all $u \in V$. 
Since the $X_u$'s are independent of each other, we can use \Cref{lemma:mcdiarmid} with $\mu = t = n$ to say that:
$\Pr[X > 2n] \le 2 \exp(n^2/n^{1 + 2\alpha}) \le 2\exp(n^{1-2\alpha}) \le 1/\poly(n).$
\end{proof}

As we observed earlier, good nodes expect to see a lot of sampled neighbors, therefore it is very unlikely that a good node has no sampled neighbors. The following lemma formalizes this intuition.

\begin{lemma} \label{lemma: good-set-aside-prob}
In Line~\ref{line:2rs-setaside-1} of \Cref{alg:ruling-set}, a good node $u$ with $\deg(u) = d$ is added to \vsetaside with probability $1/\poly(d)$. This event is independent of the randomness (for sampling into \vsampled) of nodes more than distance $1$ from $u$ in $G$.
\end{lemma}

\begin{proof}
Since $u$ is good, the sum of sampling probabilities of its neighbors is at least $\gamma\cdot \log d$. So we can bound the expected number of neighbors in \vsampled as $\E\left[|N(u)\cap \vsampled|\right] \geq \gamma\cdot \log d$.
Since the sampling is done independently for each node, we can use Chernoff bound (\Cref{lemma:ChernoffBound}) to compute the probability that no neighbor of $u$ is in \vsampled. We get that $\Pr\left[ |N(u)\cap \vsampled | = 0 \right] < 1/d^{\gamma/2}$.
Hence, $u$ is added to  \vsetaside with probability at most $1/\poly(d)$. This event only depends on the randomness of $u$ and its neighbors in $G$, therefore it is independent of the randomness of nodes at distance more than $1$ from $u$.
\end{proof}

On the other hand, bad nodes expect to have few sampled neighbors. So a sampled bad node will have a good probability of being a local minimum in the first iteration of Luby's algorithm. Therefore, nodes having many bad neighbors are very likely to have one such bad neighbor join the MIS, and hence all such bad neighbors are $2$-hop covered.

\begin{lemma}\label{lemma: bad-set-aside-prob}
In Line~\ref{line:2rs-setaside-2} of \Cref{alg:ruling-set}, each node $u \in B_d$ is added to \vsetaside with probability at most $1/\poly(d)$. This happens independently of the randomness (for sampling into \vsampled and for the first Luby round when computing \lubymis) of nodes more than distance $3$ from $u$ in $G$. 
\end{lemma}
\begin{proof}
Recall that $u$ is added to \vsetaside if there is a node $v\in N(u)$ such that $v$ has more than $c \sqrt{d}\cdot \log^5 d$ neighbors in $B_d$. Let $v$ be an arbitrary such node and let $A_u$ be a subset of $N(v) \cap B_d$ such that $|A_u| = c \sqrt{d}
\log^5 d$ and $u\in A_u$. 
We will show that at least one node of $A_u$ joins the MIS in Line~\ref{line:2rs-luby} of \Cref{alg:ruling-set} in the first Luby round with probability at least $1 - 1/\poly(d)$.

A node $w \in A_u$ joins \lubymis in the first Luby round iff $w$ is sampled and if $w$ has the smallest random number (in the first Luby round) among all its sampled neighbors. 
Note that if one node of $A_u$ joins the MIS only depends on the randomness of the bad nodes in $A_u \cup N(A_u)$, which are a subset of the $2$-hop neighborhood of $v$ and thus of the 3-hop neighborhood of $u$.

First, note that the number of nodes in $A_u \cup N(A_u)$ is at most $O(d^{3/2} \log^5(d))$ because $A_u \subseteq B_d$, so every node has at most $2d$ neighbors in $N(A_u)$. 
 
Let $S_u$ be the set of sampled nodes in $A_u$. 
Every node in $N(A_u)$ with at most $O(\sqrt{d} \log^2 d)$ neighbors in $A_u$ has at most $O(\log^2 d)$ neighbors in $S_u$ with probability at least $1-1/poly(d)$ (by using Chernoff bound \Cref{lemma:ChernoffBound} and then union bound over all such nodes in $N(A_u)$). 
On the other hand, every node $w\in N(A_u)$ with $\Omega(\sqrt{d} \log^2 d)$ neighbors in $A_u$ has $\Omega(\log^2 d)$ neighbors in $S_u$ in expectation. 
If $w$ is a good node, it does not participate in the first Luby round carried out by the nodes in $A_u$, and if $w$ is a bad node, the expected number of sampled neighbors of $w$ is at most $\gamma \log(\deg(w))$ and we therefore have $\log(\deg(w)) = \Omega(\log^2 d)$ and thus $\deg(w) = \exp(\Omega(\log^2 d))$. 
The probability that $w$ is sampled is therefore $\ll 1/poly(d)$. 

With probability $1 - 1/\poly(d)$, all the sampled bad nodes in $N(A_u)$ therefore have at most $O(\log^2 d)$ sampled neighbors in $A_u$. Moreover, all the sampled nodes in $A_u$ have at most $O(\log d)$ overall sampled neighbors, since $A_u$ is a subset of $B_d$, it has at most $\gamma \log 2d$ sampled neighbors in expectation. Using a Chernoff bound (\Cref{lemma:ChernoffBound}) gives us that each node in $A_u$ has $O(\log d)$ sampled neighbors with probability $1-1/poly(d)$. 
In the following, we condition on this event happening. 

Consider the graph $G_S$ induced by the sampled nodes in $A_u \cup N(A_u)$. 
For any two nodes $x, y \in S_u$ that are at distance at least $3$ in $G_S$, the events that $x$ and $y$ join the MIS \lubymis in the first Luby step are independent. 
For every node $x \in S_u$, there are at most $O(\log^3 d)$ other nodes in $S_u$ at distance at most 2 in $G_S$ (at most $O(\log d)$ direct neighbors and because the direct neighbors can be in $N(A_u)$ at most $O(\log^3 d)$ $2$-hop neighbors). 
By greedily picking nodes in $S_u$, we can therefore find a set of size $\Omega(|S_u|/\log^3 d)$ of nodes in $S_u$ that independently join the MIS \lubymis in the first Luby step. 
Because with probability $1 - 1/\poly(d)$, $|S_u| = \Omega(\log^5 d)$, we have $\Omega(|S_u|/\log^3 d) = \Omega(\log^2 d)$. 

Each of those nodes independently joins \lubymis with probability at least $1/O(\log d)$ and therefore, one of those nodes joins \lubymis with probability at least $1 - 1/\poly(d)$.
\end{proof}

We use the fact that nodes are added mostly independently and with low probability to \vsetaside in order to show that the graph induced by these nodes cannot have many edges.

\begin{lemma}\label{lemma: set-aside-small}
    The induced subgraph $G[\vsetaside]$ has $O(n)$ edges \whp
\end{lemma}
\begin{proof}
A node $v$ is placed in $\vsetaside$ if either (1) $v$ is a good node with no neighbors in \vsampled, or (2) $v$ is an uncovered bad node in $B_d$ and has a neighbor $v$ that has at least $c\cdot\sqrt{d}\cdot\log^5(d)$ neighbors in $B_d$. 
By Lemmas~\ref{lemma: good-set-aside-prob} and~\ref{lemma: bad-set-aside-prob}, each such node $v$ is put in $\vsetaside$ with probability at most $1/\poly(\deg(v))$, and this happens independently of the randomness (for sampling into \vsampled and for the first Luby round when computing \lubymis) of nodes at distance more than $3$ from $v$.
The exponent of the polynomial depends on $c$ and $\gamma$.

Nodes with constant degree can be ignored, as they will contribute at most $O(n)$ edges. 
Therefore, we can assume that each node is put in \vsetaside with probability at most $1/2$.

For the sake of analysis, we compute a greedy coloring of $G^7[\vsetaside]$. 
To get $G^7[\vsetaside]$, we first build the graph $G^7$ which is the graph where we add an edge between any pair of nodes that are at distance at most $7$ in $G$, and then we take the induced subgraph of $G^7$ on $\vsetaside$. 

For each color class that has at least $n^{1-8\alpha}$ nodes, all of which are at distance more than $7$ from each other in $G$, they join \vsetaside independently of each other. 
Therefore, the probability that all nodes in a single color class belongs to \vsetaside is at most $(1/2)^{n^{1-8\alpha}} \ll 1/\poly(n)$. 
By union bounding over the color classes, we get that with probability $1- 1/\poly(n)$, the size of each color class is less than $n^{1-8\alpha}$.

Each node in $G^7[\vsetaside]$ has degree at most $n^{7\alpha}$ and hence there are $n^{7\alpha} + 1$ color classes. Recall that $\Delta \le n^\alpha$ in $G$, so if a color class $C \subseteq \vsetaside$ has less than $n^{1-8\alpha}$ nodes, the number of edges in $G$ that are incident on $C$ is at most $n^{1-7\alpha}$. Therefore, all color classes with less than $n^{1-8\alpha}$ nodes can only add at most $O(n)$ edges to $G[\vsetaside]$.

The lemma follows since we already showed that \whp, the size of each color class is less than $n^{1-8\alpha}$.
\end{proof}

Recall that $B_{d}$ is the set of bad nodes in $G$ with (initial) degree in the range $[d, 2d)$.
Let $B^{*}_{d}$ be the nodes in $B_{d}$ that have a neighbor $v$ with more than $c\sqrt{d}\log^5 d$ neighbors in $B_{d}$. If a node in $B^{*}_{d}$ is not covered by the MIS \lubymis on \vsampled, then it is put into $\vsetaside$. Let $\overline{B}_{d} = B_{d} \setminus B^{*}_{d}$. Define $B = \cup_{i=0}^{\log n} B_{2^i}$ and $\overline{B} = \cup_{i=0}^{\log n}\overline{B}_{2^i}$. 

\begin{lemma} \label{lemma: residual-only-bad}
The graph $G' =  G[V \setminus (\mathcal{I} \cup  N(\mathcal{I}) \cup N(N(\mathcal{I})))]$ of nodes that are uncovered before Line~\ref{line:2rs-second-time} contains only nodes in $\overline{B}$.
\end{lemma}
\begin{proof}
    Nodes not in $\overline{B}$ are the good nodes and the bad nodes in $B^*_d$ for all $d$.
    We show that all these nodes are covered before Line~\ref{line:2rs-second-time} and hence cannot belong to $G'$.
    Nodes that are either in \vsampled, or have a neighbor in \vsampled are covered by the MIS \lubymis computed in Line~\ref{line:2rs-luby}. 
    Good nodes that are not covered by \lubymis are put in \vsetaside in Line~\ref{line:2rs-setaside-1}. 
    Similarly, nodes in $B^*_d$ that are not covered by \lubymis are put in \vsetaside in Line~\ref{line:2rs-setaside-2}.
    All nodes in \vsetaside are covered because we compute an MIS on $G[\vsetaside]$ in Line~\ref{line:2rs-setaside-mis}.
    Therefore, all nodes not in $\overline{B}$ are covered before Line~\ref{line:2rs-second-time}.
\end{proof}

\subsection{Counting the Bad Nodes}
Let $V_{\geq d}$ be the set of all nodes in $G$ of (initial) degree at least $d$. Intuitively, we now want to say: (1) for each bad node in $\overline{B}_d$, there are at least $d/2$ edges to higher degree nodes and (2) from the higher degree nodes, only roughly $\sqrt{d}$ edges to $\overline{B}_d$. Hence, we can conclude that $d \cdot |\overline{B}_d| \leq |V_{\geq d^2}| \cdot \sqrt{d}$ which further implies that $|\overline{B}_d| \leq |V_{\geq d^2}| / \sqrt{d}$.

\begin{lemma}\label{lemma: bad-inc-lower}
 Consider a bad node $u$ with $\deg(u) = d$. Then for at least $d/2$ nodes $v \in N(u)$ it holds that $\deg(v) \geq d^2/(4\gamma^2\log^2 d)$.
\end{lemma}
\begin{proof}
Otherwise, more than half of the neighbors have degree less than $d^2/(4\gamma^2\log^2 d)$. 
Hence, 
\begin{align*}
    \sum_{v \in N(u)} \frac{1}{\sqrt{\deg(v)}} \geq \frac{d}{2} \cdot \frac{\sqrt{4\gamma^2 \log^2 d}}{\sqrt{d^2}} = \frac{d}{2} \cdot \frac{2\gamma \log d}{d} = \gamma \log d \ ,
\end{align*}
which is a contradiction with $u$ being bad.
\end{proof}

\begin{lemma}\label{lemma: bad-deg-fixed}
For any $d$, we have that $|\overline{B}_{d}| \leq 2|V_{\geq d^2/(4\gamma^2\log^2 d)}|\cdot \log^5 d / \sqrt{d} \leq 2n\log^5 d /\sqrt{d}$.
\end{lemma}
\begin{proof}
Let $d' = d^2/(4\gamma^2\log^2 d)$. From \Cref{lemma: bad-inc-lower}, we know that for any $u \in \overline{B}_{d}$, at least $d/2$ edges go to $V_{\geq d'}$.
Furthermore, we have that any $v\in V_{\geq d'}$ has at most $c \sqrt{d} \log^5 d$ edges to $\overline{B}_{d}$, since otherwise, none of $v$'s neighbors are in $\overline{B}_{d}$.
Hence, we can conclude that $\frac{d}{2} \cdot |\overline{B}_{d}| \leq |V_{\geq d'}| \cdot \sqrt{d} \log^5 d$ which proves the lemma.
\end{proof}

Now we have that just before Line~\ref{line:2rs-second-time}, $|\overline{B}_d| \approx n/\sqrt{d}$, and we now run the entire algorithm again on the uncovered graph $G'$. For $G'$, we define the sets $V'_{\ge d}$, $B'_{d}$, $\overline{B'}_{d}$, $B'$, and $\overline{B'}$ similarly as we did for $G$. Now we would expect that $|\overline{B'}_d| \approx n/d$, which is good because all nodes have degree at most $2d$. 

\begin{lemma}\label{lemma: bad-all-deg}
The graph induced by the uncovered nodes before Line~\ref{line:2rs-mis-residual} has $O(n)$ edges.
\end{lemma}
\begin{proof}
Again let $d' = d^2/(4\gamma^2\log^2 d)$. By a similar argument as \Cref{lemma: residual-only-bad}, the uncovered nodes are a subset of $\overline{B'}$.
We can assume that $d$ is at least some large enough constant, as the nodes in all $B'_{d}$ for constant $d$ have at most $O(n)$ edges.
By \Cref{lemma: bad-deg-fixed}, we have that $|\overline{B}_{d}| \leq 2|V_{\ge d'}|\log^5 d /\sqrt{d}$ for any $d$, and we can similarly argue $|\overline{B'}_{d}| \leq 2|V'_{\ge d'}|\log^5 d /\sqrt{d}$ for any $d$.

Now $V'_{\ge d'} \subseteq \overline{B}_{\ge d'}$ where $\overline{B}_{\ge d'} = \cup_{i=\log d'}^{\log n} \overline{B}_{2^i}$. Therefore,
\begin{align*}
|V'_{\ge d'}| \leq \sum_{i=\log d'}^{\log n} \frac{2 n\log^5 2^i}{\sqrt{2^{i}}} = \sum_{i=\log d'}^{\log n} \frac{2n i^5}{2^{i/2}} \le \sum_{i=\log d'}^{\log n} \frac{2n}{2^{i/3}} \le O\left(\frac{n \log^{2/3} d}{d^{2/3}}\right)
\end{align*}

Where the last inequality follows because the sum on the left is a geometric sequence with rate $2^{-1/3}$ and it is well known that if the rate is between $0$ and $1$, the sum is asymptotically dominated by the first term. Therefore, $|\overline{B'}_{d}| \leq O(n \cdot (\log^6 d) / d^{7/6}) \le O(n/d^{1.1})$ since we assumed that $d$ is a large enough constant\footnote{The conference version bounds $|\overline{B'}_{d}|$ more loosely by $O(n/d)$ which is not enough to prove the lemma}.
Since each node in $\overline{B'}_{d}$ has degree at most $2d$, the graph induced by nodes in $\overline{B'}$ has edges bounded by $\sum_{d} O(n/d^{0.1}) \le \sum_{i \ge 0} O(n/2^{i/10}) = O(n)$.
\end{proof}

\subsection{Degree Reduction}\label{sec:degree-reduction}

We define the degree reduction process similar to \cite{ghaffari2018improved} and for sake of completeness we provide a self-contained explanation here. 
Our goal is to lower the maximum degree of the input graph $G$ such that, after a constant number $i$ of steps, the maximum degree is strictly less than $n^{\alpha}$ for some fixed constant $\alpha < 1/8$.
Let $G_1 = G$, and $\Delta_{j} = \Delta^{(3/4)^j}$. Let $i$ be a value such that for all $1 \le j<i$, $\Delta_j > n^{\alpha/2}$ and $\Delta_i \le n^{\alpha/2}$. Since $\Delta \le n$, we can say that the largest value $i$ can take is $\lceil \log_{4/3} (2/\alpha)\rceil = O(1)$.

In each step $j = 1 \dots i$, we sample nodes $S_j$ in $G_{j}$ with probability $1/\Delta_j$, and then compute an MIS on the subgraph induced by the sampled nodes $G_{j}[S_{j}]$. The residual graph $G_{j+1}$ is obtained by removing all the neighbors of the sampled nodes.  
For each of these steps to be possible, the graph induced by the sampled nodes must have $O(n)$ edges \whp 
In order to guarantee the feasibility of step $j$, the maximum degree in the residual graph after step $j-1$ (i.e. $G_j$) has to be sufficiently small, which we show in the following lemma.
\begin{lemma}\label{lemma:maxDegInPrefixGraph}
    If we process a graph $G_{j}$ induced by nodes picked uniformly at random with probability $1/\Delta_j$, the maximum degree in the residual graph $G_{j+1}$ is $O(\Delta_j\log n)$ \whp
\end{lemma}

\begin{proof} 
    Consider a node $v$ in the residual graph such that $\deg (v) > d$. 
    The probability that a neighbor of $v$ is sampled is at least $1/\Delta_j$. 
    Therefore, the probability that no neighbor of $v$ is sampled is at most $\left( 1-\frac{1}{\Delta_j} \right)^d\leq \exp{(-d/\Delta_j)}$. 

    Denote $c>1$ an arbitrary constant, and suppose that $d = c\Delta_j\log n$. 
    Then,  the probability that $\deg(v)> d$ is at most $\exp{(-c \log n)}= n^{-c}$. 
    Hence, $\deg(v)=O(\Delta_j\log n)$ \whp 
    We conclude the lemma by union bounding over all nodes of the residual graph. 
\end{proof}

Therefore, the residual graph $G_{i+1}$ has maximum degree $O(\Delta_i\log n)$ $\le O(n^{\alpha/2}\log n) $ $\le n^\alpha$. This guarantee is our assumption in \Cref{alg:ruling-set} and it holds \whp We now finish by showing that each induced subgraph has linear size.

\begin{lemma}
    For all $1 \leq j \leq i$, the graph induced by sampled nodes in step $j$, $G_j[S_j]$, has $O(n)$ edges \whp
\end{lemma}

\begin{proof}

    First, consider a node $v\in S_j$, and $u\in N(v)$. The probability that $u\in S_j$ is $1/\Delta_j$. By \Cref{lemma:maxDegInPrefixGraph}, we condition on the high probability event that $O\left(\Delta_{j-1}\log n\right)$ is the maximum degree of $G_j$. Note that this conditioning only affects the randomness used for sampling in iterations $1, \dots, j-1$. In particular, this implies that the conditioning does not affect the randomness used for sampling in iteration $j$. Therefore, the expected degree of $v$ in $S_j$ is at most $\mu = O\left(\log n \cdot \Delta_{j-1}/\Delta_{j}\right)$.
    Using \Cref{lemma:ChernoffBound}, $\Pr[\deg(v)\geq (1+c)\mu]\leq \exp(-c^2\mu /(2+c)) \leq n^{-c}$, since $\mu \ge \log n$.
    By union bounding over all nodes of $S_j$, the maximum degree in $G_j[S_j]$ is $O\left(\log n \cdot \Delta_{j-1}/\Delta_{j}\right)$ \whp 

    Second, the expected number of nodes in $S_j$ is at most $\mu' = n/\Delta_{j}$. Using \Cref{lemma:ChernoffBound}, $\Pr[|S_j| \geq (1+c\log n)\mu']\leq \exp(-c^2 \log^2 n\mu' /(2+c\log n)) \leq n^{-c}$, since $\mu' \ge 1$. Therefore, $|S_j| = O(n\log n/\Delta_{j})$ \whp
    
    Now we can upper bound the number of edges in $G_j[S_j]$ by $|S_j|\cdot \Delta(G_j[S_j])$. Therefore, we can say \whp that 
    $$ |S_j|\cdot \Delta(G_j[S_j]) = O\left(\frac{n \log n}{\Delta_{j}} \cdot \frac{\Delta_{j-1}\log n}{\Delta_{j}} \right)  = O\left(n\log^2 n \cdot \frac{\Delta_{j-1}}{\Delta_{j}^2}\right)$$ 
    Moreover, since $\Delta_{j} = \Delta^{(3/4)^j}$, we have that $\Delta_{j-1}/\Delta_{j}^2 = \Delta^{(3/4)^{j-1} - 2 (3/4)^{j}} = \Delta^{-(1/2)(3/4)^{j-1}} = 1/\sqrt{\Delta_{j-1}}$. This term cancels the $\log^2 n$ term since $\Delta_{j-1} > n^{\alpha/2} \gg \log^4 n$. Hence, the number of edges in $G_j[S_j]$ is $O(n)$ \whp

\end{proof}

\section{Implementation of \parallelrulingset}\label{sec: implementation-main}
In this section, we show how to implement the algorithm \parallelrulingset in several models of parallel computation. In each case of the implementations, we only need to bound the runtime and memory usage of the algorithm in the corresponding model. Since we faithfully execute the \parallelrulingset algorithm, the solutions computed are guaranteed to be correct.

\subsection{\clique}
In the \clique model~\cite{LotkerPPP2003CongClique}, we have $n$ machines, where each machine is identified with a single node in the input graph. The communication network is a clique, that is, the machines are connected in all-to-all fashion, and the input graph is considered to be a subgraph of the network.
Machines can send unique messages to all other machines via the edges in the clique, and the bandwidth of each edge is limited to $O(\log n)$ bits. \clique and \mpc are closely related to each other, for example \cite{HegemanPTCS2015, BehnezhadDH2018} show how to implement any \clique algorithm in the \mpc model.

By the definition of \parallelrulingset, we compute an MIS sequentially on several subgraphs: (1) the sampled graphs $G_j[S_j]$ ($1 \le j \le i = O(1)$) for reducing the degree, (2) the graphs induced by $\vsampled$ and $\vsetaside$ on $G$ during the first run and on $G'$ during the second run, and finally (3) the graph induced by uncovered nodes in the last step. Creating a subgraph takes a constant number of rounds, since nodes just need to know the random choices and aggregate information of their $1$-hop neighbors. 

By the lemmas in \Cref{subsection:2rs-small-graphs} and \Cref{sec:degree-reduction}, all these subgraphs have $O(n)$ edges \whp Therefore, we can use Lenzen's routing protocol~\cite{Lenzen2013} to gather all the subgraphs one after the other at a single machine in $O(1)$ rounds. This machine computes the MIS according to \Cref{alg:ruling-set}, and informs the rest of the nodes whether they joined the MIS or not, which allows us to identify the next subgraph. Therefore, we get the following result.

\begin{theorem}\label{thm: 2rs-congestedClique}
    There is a Congested Clique algorithm to find a $2$-ruling set. The algorithm runs in $O(1)$ rounds \whp
\end{theorem}

\subsection{Linear Memory \mpc Model}

In the \mpc model~\cite{Karloff2010, Goodrich2011, Beame2017}, we have $M$ machines with $S$ words of memory each, where each word corresponds to $O(\log n)$ bits.
Notice that an identifier of an edge or a node requires one word to store.
The machines communicate in an all-to-all fashion.
The input graph is divided among the machines and for simplicity and without loss of generality, we assume that the edges of each node are placed on the same machine.
In the \emph{linear-space} \mpc model, we set $S = \Theta(n)$. 
Furthermore, the \emph{total space} is defined as $M \cdot S$ and in our case, we have $M \cdot S = \Theta(m)$.
Notice that $M \cdot S = \Omega(m)$, for the number of edges $m$ in the input graph, simply to store the input.

The implementation follows as a direct consequence of the \clique model. Since the local memory is $\Theta(n)$, we can send each of the $O(n)$ sized subgraphs one by one to a single machine in $O(1)$ rounds and use this machine to compute the MIS as described in the algorithm. We obtain the following theorem.

\begin{theorem}\label{thm: 2rs-linearMPC}
    There is an \mpc algorithm with $O(n)$ words of local and $O(m)$ total memory to find a $2$-ruling set.
    The algorithm runs in $O(1)$ rounds \whp
\end{theorem}

\subsection{Semi-streaming}
Typically, in the distributed and parallel settings, the input graph is too large to fit a single computer.
Hence, it is divided among several computers (in one way or another) and the computers need to communicate with each other to solve a problem.
Another angle at tackling large datasets and graphs is through the \emph{graph streaming models}~\cite{Feigenbaum2005, Feigenbaum2005a, Mut2005}.
In these models, the graph is not stored centrally, but an algorithm has access to the edges one by one in an input stream, chosen randomly or by an adversary (the choice sometimes makes a difference).
We assume that each edge is processed before the next pass starts.
In the semi-streaming setting, the algorithm has $\widetilde{O}(n)$ working space, that it can use to store its state.
The goal is to make as few \emph{passes} over the edge-stream as possible, ideally just a small constant amount.
Notice that in the case of many problems, such as matching approximation or correlation clustering, simply storing the output might demand $\Omega(n)$ words.

In the semi-streaming model, we use one pass to process one subgraph of size $O(n)$, by storing it in memory and computing an MIS as described in the algorithm. Since there are $O(1)$ such subgraphs, we require $O(1)$ passes. This leads to the following theorem.

\begin{theorem}\label{thm: 2rs-streaming}
    There is an $O(1)$-pass semi-streaming algorithm with $O(n)$ words of space to find a $2$-ruling set \whp
\end{theorem}

\section{Concentration Inequalities}\label{sec: concentration}
\begin{lemma}[Chernoff Bounds]
\label{lemma:ChernoffBound}
Let $X_1,\ldots,X_k$ be independent $\{0,1\}$ random variables. Let $X$ denote the sum of the random variables, $\mu$ the sum's expected value. Then,
\begin{enumerate}
    \item For $0 \leq \delta \leq 1$, $\Pr[X \leq (1-\delta) \mu] \leq \exp(- \delta^2 \mu / 2)$ and $\Pr[X \geq (1+\delta) \mu] \leq \exp(- \delta^2 \mu / 3)$,
    \item For $\delta \geq 1$, $\Pr[X \geq (1+\delta) \mu] \leq \exp(- \delta^2 \mu / (2+\delta))$.
\end{enumerate}
\end{lemma}
\begin{definition}[Bounded Differences Property]
    A function $f: \mathcal{X} \to \mathbb{R}$ for $\mathcal{X} = \mathcal{X}_1 \times \mathcal{X}_2 \dots \times \mathcal{X}_n$ is said to satisfy the bounded differences property with bounds $c_1, c_2, \dots, c_n \in \mathbb{R}^{+}$ if for all $\overline{x} = (x_1, x_2, \dots, x_n) \in \mathcal{X}$ and all integers $k \in [1, n]$ we have
    $$\sup_{x'_k \in X_k}{|f(\overline{x}) - f(x_1, x_2, \dots, x_{i-1}, x'_k, \dots, x_n)|} \le c_k$$
\end{definition}

\begin{lemma}[Bounded Differences Inequality \cite{mcdiarmid1989method,mcdiarmid1998concentration}]
    \label{lemma:mcdiarmid}
    Let $f: \mathcal{X} \to \mathbb{R}$ satisfy the bounded differences property with bounds $c_1, c_2, \dots, c_n$. Consider independent random variables $X_{1},X_{2},\dots ,X_{n}$ where $X_{k}\in \mathcal {X}_{k}$ for all integers $k \in [1,n]$. Let $\overline{X} = (X_{1},X_{2},\dots ,X_{n})$ and $\mu = \E[f(\overline{X})]$. Then for any $t>0$ we have:
    $$\Pr[|f(\overline{X}) - \mu| \ge t] \le 2 \exp\left(\frac{-t^2}{\sum_{k=1}^n c_k^2}\right)$$
\end{lemma}

\bibliographystyle{plain}
\bibliography{main.bib}


\end{document}